\documentclass[letterpaper]{article} 
\usepackage{aaai2027}  
\usepackage[hyphens]{url}  
\usepackage{graphicx} 
\usepackage{natbib}  
\usepackage{caption} 
\usepackage{amsmath, amssymb, amsthm, bm}
\usepackage{algorithm}
\usepackage{algorithmic}

\usepackage{newfloat}
\usepackage{listings}
\DeclareCaptionStyle{ruled}{labelfont=normalfont,labelsep=colon,strut=off} 
\floatstyle{ruled}
\newfloat{listing}{tb}{lst}{}
\floatname{listing}{Listing}

\newtheorem{theorem}{Theorem}
\newtheorem{lemma}{Lemma}
\newtheorem{corollary}{Corollary}
\newtheorem{definition}{Definition}

\usepackage{booktabs}

\newcommand{\cf}{^{-z}}   

\title{Forgetting Without Restarting:\\Execution-State Unlearning for Stateful LLM Agents}
\author{
    Written by AAAI Press Staff\textsuperscript{\rm 1}\thanks{With help from the AAAI Publications Committee.}\\
    AAAI Style Contributions by Peter Patel Schneider,
    Sunil Issar,\\
    J. Scott Penberthy,
    George Ferguson,
    Hans Guesgen,
    Francisco Cruz\equalcontrib\corresponding,
    Marc Pujol-Gonzalez\equalcontrib\corresponding
}
\affiliations{
    \textsuperscript{\rm 1}Association for the Advancement of Artificial Intelligence\\

    1101 Pennsylvania Ave, NW Suite 300\\
    Washington, DC 20004 USA\\
    proceedings-questions@aaai.org
}

\author{
    Chao Yao\textsuperscript{\rm 1},
    Yangbo Wei\textsuperscript{\rm 2},
    Zhen Huang\textsuperscript{\rm 2},
    Junhong Qian\textsuperscript{\rm 2},\\
    Chenle Chen\textsuperscript{\rm 2},
    Shaoqiang Lu\textsuperscript{\rm 2},
    Chen Wu\textsuperscript{\rm 2},
    Lei He\textsuperscript{\rm 2}\corresponding
}
\affiliations{
    \textsuperscript{\rm 1}Arizona State University, USA\\
    \textsuperscript{\rm 2}Eastern Institute of Technology, Ningbo, China\\
    cyao22@asu.edu,
    yangforever@sjtu.edu.cn
}

\begin{document}

\maketitle

\begin{abstract}
Long-running LLM agents are stateful: beyond the transcript they accrete compressed summaries, plaintext memory, pending tool plans, and---under every serving API---a KV cache. Yet today's ``forget'' operations delete a plaintext memory record and stop, leaving every artifact derived from the revoked information intact. We formalize \emph{execution-state unlearning}: after a forget request, the agent must behave as if it had never observed the target. Modeling the runtime as a deterministic transition system, we prove that the pre-target trajectory prefix is shared with this counterfactual world for free, that the post-target suffix is irreducibly tainted without token-level attribution, and that exact unlearning requires at least $T-\tau+1$ recomputed transitions, where $\tau$ is the target's injection step. \emph{Provenance-Guided Selective Replay} attains this bound as a cross-layer contract spanning prompt, compressed memory, and cache: a provenance graph locates the injection point, checkpoint restoration reduces to \emph{cropping} the KV cache, and sanitized replay regenerates the counterfactual suffix. Audited with elicitation, stochastic, and string-free behavioral tests across three agent suites, nine baselines, and three model families, memory deletion leaves leakage unchanged, instruction-based forgetting collapses under elicitation (Leak@probes $=1.00$), and source redaction still \emph{acts} on a revoked preference in 80\% of episodes---while selective replay is indistinguishable from a full reset at up to $9\times$ fewer recomputed tokens.
\end{abstract}

\section{Introduction}

\begin{figure}[t]
\centering
\includegraphics[width=0.97\columnwidth]{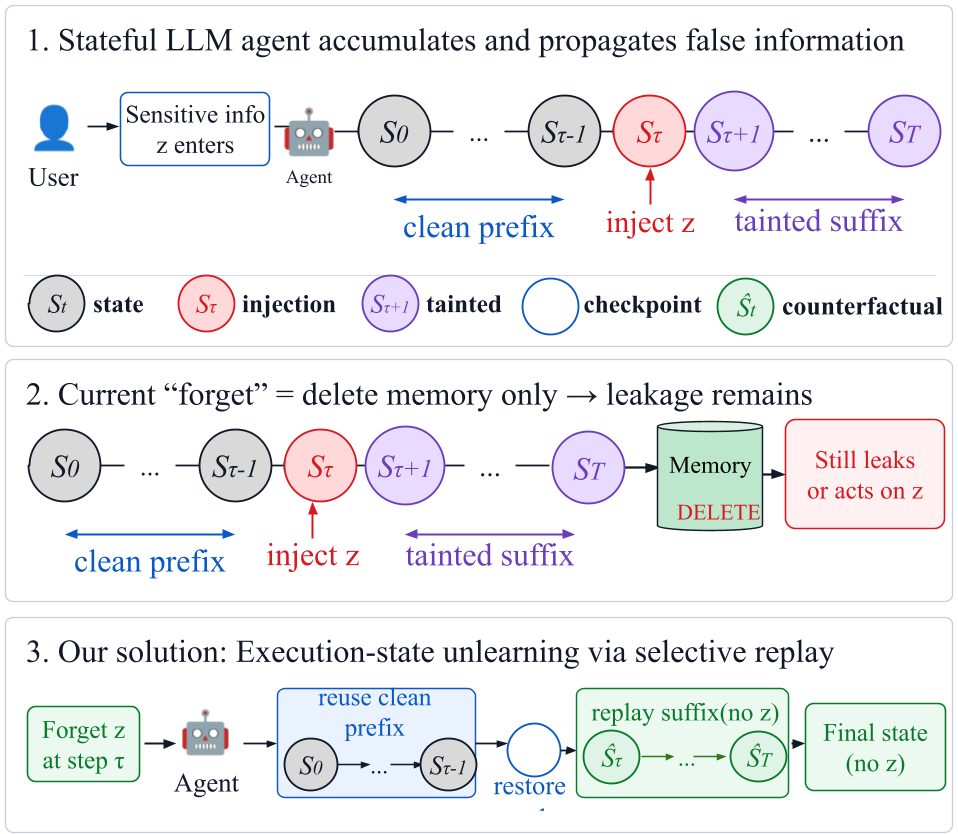}
\caption{Execution-state unlearning at a glance. A target $z$ entering at step $\tau$ splits the trajectory into a clean prefix and a suffix whose summaries, plans, and cache are all tainted \emph{(1)}. Deleting the memory record leaves that suffix intact---the agent still leaks or acts on $z$ \emph{(2)}. Selective replay restores $\hat S_{\tau-1}$ by \emph{cropping} the cache, then replays the sanitized suffix to $\hat S_T$ \emph{(3)}.}
\label{fig:motivation}
\end{figure}

Within months of release, agent frameworks such as OpenClaw \cite{openclaw2026} and Hermes Agent \cite{hermes2026} accumulated hundreds of thousands of deployments as \emph{always-on personal agents} that run for weeks, operate tools, and remember their users. What makes them useful is precisely that they are \emph{stateful}: a modern runtime layers the transcript; \emph{context compaction} that rewrites older turns into model-authored summaries; plaintext long-term memory re-injected at session start \cite{packer2024memgpt,chhikara2025mem0}; tool traces and pending plans; and, beneath all of these, the KV cache---universal serving infrastructure reused across requests by every engine and commercial API \cite{kwon2023vllm,zheng2024sglang,gim2024promptcache}. Whatever enters an agent's context is compressed into summaries, distilled into plans, persisted into memory, and materialized as cached tensors (Figure~\ref{fig:motivation}). This collides with an equally basic requirement: sometimes the agent must \emph{un-see} something---a user revokes consent for a home address mid-session \cite{gdpr2016}, a secret is pasted by accident, or an indirect injection plants content inside a fetched page or tool result \cite{greshake2023not}. The last case is the common one and the user never witnesses it: the agent silently reads the injected content, folds it into its state, and moves on; the forget trigger, when it comes, comes from a detector or operator \emph{after the fact}. Yet deployed stacks offer only a forgetting affordance that operates on \emph{plaintext at a single layer}: delete the memory record, edit the Markdown file, drop the message from retrieval. The industry has equated forgetting with un-indexing.

We show systematically that this equation fails, and that the failure is invisible to the string-matching evaluations used to certify it. On three agent suites instrumented with memory injection, compaction, and tool use \cite{wu2025longmemeval,lu2024toolsandbox,debenedetti2024agentdojo}, deleting the persistent memory record leaves leakage \emph{exactly} unchanged from doing nothing (0.86--1.00 any-leak): the target survives in the session's derived state. An instruction to forget looks far better on a single task-shaped probe, but under a six-probe elicitation audit the same state yields the target with probability $1.00$---merely suppressed. Source redaction fails more subtly: the model's own summary re-encodes the target in paraphrase, beyond any forbidden-string list, and in a behavioral test the redacted agent still \emph{acts} on a revoked preference in 80\% of episodes while emitting the string zero times. String metrics certify precisely the methods that fail. Even information-flow control \cite{costa2025fides}, which blocks tainted \emph{future} flows, cannot clean state already contaminated: an IFC-only baseline leaks at the no-forget rate.

What should ``forget'' mean for a running agent? We argue for a counterfactual criterion: future behavior must be indistinguishable from a twin agent that \emph{never observed} the target. Modeling the runtime as a deterministic transition system, we define the counterfactual trajectory induced by deleting the target $z$ from the observation stream at its injection step $\tau$, and call an operator an exact \emph{execution-state unlearner} if it maps the real final state to the counterfactual one; unlike parameter unlearning \cite{cao2015towards,bourtoule2021machine,maini2024tofu}, the edited object is non-parametric runtime state, where exactness is attainable. The formalism yields sharp structure: a prefix-sharing lemma makes the first $\tau{-}1$ counterfactual steps \emph{free} (the clean prefix is literally a prefix of the contaminated cache, so restoring it is a crop); a taint lemma shows that without token-level attribution every artifact at or after $\tau$ is unsalvageable---computation cannot be edited, only replayed. A splicing theorem then proves checkpoint-and-replay reconstructs the counterfactual state exactly, with a matching $T-\tau+1$ lower bound: forgetting cost is governed by the counterfactual divergence, not the session length.

We realize this as \emph{Provenance-Guided Selective Replay}, an auditable cross-layer contract from prompt to compressed memory to KV cache: an artifact-level provenance graph recorded during execution, sparse metadata-only checkpoints whose restoration is a cache crop, and sanitized replay with shadow-executed, deduplicated consequential tools. Because string matching cannot certify forgetting, we audit with \emph{Leak@probes} (six elicitation probes), \emph{Leak@5} (stochastic samples), a string-free \emph{behavioral-extraction} suite, and counterfactual-action divergence, each anchored to a measured false-positive floor. Selective replay sits at the floor on every axis while recomputing up to $9\times$ fewer tokens than a full reset, with cost tracking the proven $T-\tau+1$ line ($R^2\!\approx\!1$); results reproduce across Llama-3.1-8B, Qwen2.5-7B, and Mistral-7B.

Our contributions: \textbf{(1) Problem}: execution-state unlearning for stateful LLM agents, formalized via counterfactual equivalence over reconstructible runtime state---the layer today's forget operations silently skip. \textbf{(2) Theory}: the prefix-sharing and certifiable-taint-boundary lemmas, an exact splicing theorem, and $T-\tau+1$ optimality. \textbf{(3) System}: Provenance-Guided Selective Replay, composing provenance, crop-as-restore checkpoints, and side-effect-safe replay into one forgetting contract. \textbf{(4) Audits and evidence}: elicitation, stochastic, and string-free behavioral audits with measured floors; nine baselines on three suites and three model families---every deployed-style forget fails at least one audit; selective replay matches a full reset at a fraction of its cost.

\section{Related Work}

\paragraph{Machine unlearning.}
Unlearning classically removes training data's influence from \emph{model parameters}, exactly by retraining from sharded checkpoints \cite{cao2015towards,bourtoule2021machine} or approximately by fine-tuning \cite{eldan2023whos}, with benchmarks surveyed for LLMs \cite{maini2024tofu,liu2025rethinking}; approximate unlearning is notoriously hard to verify. Our setting inverts this: the object is the agent's \emph{non-parametric execution state}, whose transition function is replayable, so \emph{exact} unlearning is attainable and certifiable. Conceptually, checkpoint-and-replay is the runtime analogue of SISA's shard-and-retrain \cite{bourtoule2021machine}, except that causality gives the shard boundary ($\tau$) for free.

\paragraph{Agent memory systems.}
Long-horizon agents externalize state into managed memory: paged context \cite{packer2024memgpt}, extracted fact stores \cite{chhikara2025mem0}, and, in deployed frameworks, plaintext Markdown or SQLite memories with compaction \cite{openclaw2026,hermes2026}. All expose deletion of a stored record; none propagate it into the live session's derived artifacts or cache, and memory benchmarks \cite{wu2025longmemeval} evaluate recall, not revocation. Our episodes exercise exactly these abstraction layers, and their delete operation is baseline B1---behaviorally a no-op.

\paragraph{KV-cache reuse and serving.}
Prefix caching is universal serving infrastructure: paged attention \cite{kwon2023vllm}, radix-tree prefix sharing \cite{zheng2024sglang}, and modular attention reuse \cite{gim2024promptcache} all reuse attention states keyed on byte-identical prefixes. This machinery is built for \emph{reuse}, not \emph{revocation}: it provides no statement about what a cached suffix still encodes. We run the same mechanism in reverse---Lemma~\ref{lem:prefix} makes \emph{cropping} a cache a certified restoration operator---and add what caching cannot: a provenance-backed guarantee of what was regenerated and why.

\paragraph{Agent security: prevention, detection---and no remediation.}
Indirect prompt injection \cite{greshake2023not,debenedetti2024agentdojo} has produced two defense families. \emph{Prevention by design} constrains what untrusted content can do before it does it: instruction-hierarchy training \cite{wallace2024instruction}, quarantine and plan-then-execute patterns \cite{beurerkellner2025design}, capability policies over extracted data flows (CaMeL; \citealt{debenedetti2025camel}), and information-flow labels with deterministic sink gating (FIDES; \citealt{costa2025fides}). \emph{Detection} flags injected content via classifiers, model-internal features, or localization of the injected span \cite{jia2026promptlocate}. Both leave the same gap: prevention is imperfect, and detection is routinely \emph{asynchronous}---by the time a flag fires, the agent has already read the content, folded it into summaries, plans, and cache, and moved on. What happens to that session is unaddressed: IFC constrains future flows but cannot clean resident state (our B8 leaks at the no-forget rate), and no defense above offers a remediation primitive. We supply the recovery half: a detector's verdict is exactly the $(z,s)$ input Algorithm~\ref{alg:replay} consumes, and composing IFC with splicing (B9) enforces sink policies over a runtime that no longer contains the target.

\begin{figure*}[t]
\centering
\includegraphics[width=0.74\textwidth]{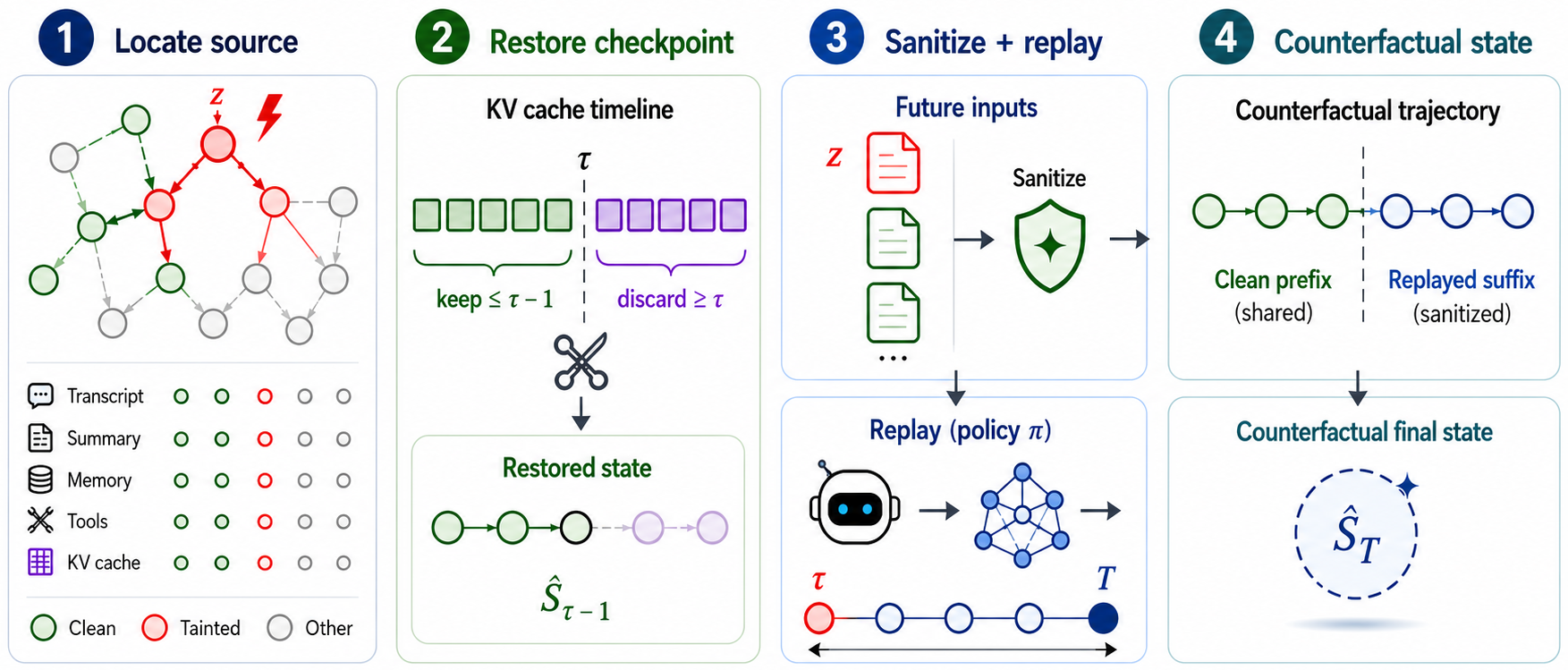}
\caption{Provenance-Guided Selective Replay (Alg.~\ref{alg:replay}). A reachability query returns $\tau$ and the taint closure \emph{(1)}; the KV timeline is cropped at $\tau{-}1$, the free prefix of Lemma~\ref{lem:prefix} \emph{(2)}; the suffix is reconstructed with $z$ dropped and all else verbatim \emph{(3)}; the splice ends at $\hat S_T = R\cf_T$ \emph{(4)}, Theorem~\ref{thm:splice}.}
\label{fig:arch}
\vspace{-10pt}
\end{figure*}

\section{Method: Forgetting as Counterfactual Trajectory Splicing}
\label{sec:method}

Formalizing the runtime as a deterministic transition system makes the ``twin agent that never saw the target'' a well-defined object---the counterfactual trajectory. Causality guarantees the two trajectories coincide pointwise before the target enters, so forgetting reduces to a \emph{splicing} problem: reuse the common prefix already computed, and re-enact only the post-divergence suffix in the world without the target. The minimal cost of forgetting is thus governed by the \emph{counterfactual divergence} $T-\tau$, not the session length $T$; our method realizes this bound as an executable system.

\subsection{The Runtime as a Deterministic Transition System}
\label{sec:transition}

Let the runtime state be $R_t \in \mathcal{R}$ (comprising the KV cache, working memory, uncommitted tool plans---all reconstructible components), and let $o_1,\dots,o_T \in \mathcal{O}$ be the external observation stream (user inputs, memory injections, tool returns). A session is the trajectory
\begin{equation}
R_{t} = F(R_{t-1},\, o_t;\ \theta), \quad t=1,\dots,T, \quad R_0 = R_{\mathrm{init}},
\end{equation}
where $F:\mathcal{R}\times\mathcal{O}\to\mathcal{R}$ is determined jointly by the model's forward computation and the agent scaffold. so transition $t$ consumes $o_t$ and produces $R_t$, and $R_T$ is the final state. The forget target $z$ first enters through an observation at step $\tau \ge 1$: $z \in o_\tau$ and $z \notin o_t$ for all $t < \tau$.

Define the \emph{counterfactual observation stream} $o\cf = (o\cf_1,\dots,o\cf_T)$ with $o\cf_t = o_t$ for $t \neq \tau$ and $o\cf_\tau = o_\tau \setminus \{z\}$. With $R\cf_0 = R_{\mathrm{init}}$ and $R\cf_{t} = F(R\cf_{t-1}, o\cf_t;\theta)$, this induces the \emph{counterfactual trajectory} $\{R\cf_t\}$---the parallel world in which the agent never saw $z$.

\begin{definition}[Counterfactual-equivalent forgetting]
\label{def:cf}
An unlearning operator $U:\mathcal{R}\times\mathcal{Z}\to\mathcal{R}$ is \emph{exact} iff $U(R_T,z) = R\cf_T$ (deterministic decoding); under stochastic decoding this relaxes to $\varepsilon$-consistency of future behavior distributions, $D\big(P_A(\cdot \mid U(R_T,z)),\ P_A(\cdot \mid R\cf_T)\big) \le \varepsilon$.
\end{definition}

The definition packages three intuitive requirements at once: the target is no longer accessible (the counterfactual never contained $z$); derived influence is removed (the counterfactual summary/plan never depended on $z$); and non-target utility is preserved (all other observations are kept verbatim). Note also what it does \emph{not} assume: who asks. The operator consumes only a revocation event $(z,s)$ naming the target and its source artifact---raised by the user, the platform, or an injection detector flagging a tool observation \emph{after} the agent has processed it; in that last, common case the user never saw $z$, and provenance, not human recollection, locates $\tau$. The operator realizing $U$ works over the recorded runtime---observation log, provenance graph, checkpoints, environment snapshots---the computational model made explicit in Corollary~\ref{cor:lb}.

\subsection{Two Lemmas: a Free Prefix and a Stubborn Suffix}

\begin{lemma}[Prefix sharing]
\label{lem:prefix}
For all $t \le \tau - 1$, $R_t = R\cf_t$.
\end{lemma}
\begin{proof}
By induction on $t$. \emph{Base}: $R_0 = R_{\mathrm{init}} = R\cf_0$. \emph{Step}: if $R_{t-1} = R\cf_{t-1}$ for some $t \le \tau-1$, then
$R_{t} = F(R_{t-1},o_t;\theta) = F(R\cf_{t-1},o_t;\theta) = F(R\cf_{t-1},o\cf_t;\theta) = R\cf_{t}$,
using the induction hypothesis and $o_t = o\cf_t$ for $t < \tau$.
\end{proof}

The proof is trivial; the corollary is not: \emph{the first $\tau-1$ counterfactual steps have already been computed, for free, by the real execution}. The clean prefix is not a cache-optimization trick---it is mathematically the shared part of the two worlds, and any scheme that discards it (e.g., a full reset) recomputes history on which the trajectories are identical.

\begin{lemma}[Taint monotonicity and the certifiable boundary]
\label{lem:taint}
Record runtime dependencies as a directed graph $G_t=(A_t,E_t)$, with $A_t$ the artifacts produced up to $t$ and $E_t$ the recorded data-flow edges; define the influence set as the reachability closure $I_t(z) = \{a \in A_t : z \rightsquigarrow_{G_t} a\}$. Then:
(i) \emph{monotonicity}: $I_t(z) \subseteq I_{t+1}(z)$ for all $t$;
(ii) \emph{certifiable boundary}: absent per-token influence attribution (i.e., without decomposing $F$ into selective reads of state components), the maximal artifact set certifiably independent of $z$ is exactly the prefix output $\{a: \mathrm{turn}(a) < \tau\}$.
\end{lemma}
\begin{proof}
(i) Execution only appends nodes and edges: $A_t \subseteq A_{t+1}$, $E_t \subseteq E_{t+1}$, and reachability is monotone in the edge set. (ii) Artifacts with $\mathrm{turn}(a)<\tau$ are generated by the shared prefix of Lemma~\ref{lem:prefix}, so independence from $z$ is directly certifiable. Conversely, any artifact with $\mathrm{turn}(a)=t\ge\tau$ is generated by a transition that reads the full state $R_t$, and $z \rightsquigarrow R_\tau \rightsquigarrow \cdots \rightsquigarrow R_t$, so the conservative graph contains a path $z \rightsquigarrow a$, i.e., $a \in I_t(z)$; excluding that path would require proving the invocation of $F$ did not use the $z$-dependent components of $R_t$---exactly the per-token attribution capability we excluded. Hence the certified-clean set is $A_t \setminus I_t(z) = \{a:\mathrm{turn}(a)<\tau\}$.
\end{proof}

Lemma~\ref{lem:taint} is the theoretical root of ``deletion $\neq$ forgetting'': once read, $z$'s influence propagates along the reachability closure into summaries, plans, and pending tool calls; local edits can remove nodes of $I(z)$ but cannot reverse computation that has already happened---\emph{computation cannot be edited, only replayed}. Here, non-reusability refers to the \emph{original} post-target KV and model-derived runtime states: no state at or after $\tau$ may be carried over. It does not mean that every post-target token must be re-\emph{decoded}. Content whose value is fixed in the counterfactual world---recorded observations under Assumption~(A2), and, when an attribution oracle stronger than the one Lemma~\ref{lem:taint}(ii) assumes away certifies it, model turns independent of $z$---may be re-materialized by prefill on the freshly reconstructed cache; this is still reconstruction, since the original post-target KV representation is never reused.

Lemma~\ref{lem:taint} also separates two kinds of selectivity. Direct state \emph{reuse} is confined to the time dimension: only the clean prefix survives, and there is no per-item triage of the suffix's KV. \emph{How} each reconstructed transition is recomputed is a separate question: counterfactually fixed content can be replayed by prefill, while genuinely model-derived content must be decoded again.

\subsection{The Splicing Theorem and Optimality}

Lemma~\ref{lem:prefix} says the clean prefix can be reused directly; Lemma~\ref{lem:taint} says the original suffix state cannot be retained and must be reconstructed in order. This state-level reconstruction does not require autoregressively regenerating every recorded token: content fixed under Assumption~(A2) may be replayed through prefill. Their combination is the method:

\begin{theorem}[Splicing equivalence]
\label{thm:splice}
Let a checkpoint exist at $\tau-1$ (or any earlier clean boundary). Define the replayed trajectory $\tilde R_{\tau-1} = R_{\tau-1}$, $\tilde R_{t} = F(\tilde R_{t-1},\tilde o_t;\theta)$ for $t \ge \tau$, with sanitized observations $\tilde o_t$. If
\emph{(A1)} decoding is deterministic (or the randomness source is fixed);
\emph{(A2)} sanitized observations agree with the counterfactual ones, $\tilde o_t = o\cf_t$ for all $t \ge \tau$; and
\emph{(A3)} no committed external side effects exist after $\tau$ (the environment can be restored from a snapshot so the tool observations in (A2) are reproducible);
then $\tilde R_t = R\cf_t$ for all $t \ge \tau-1$; in particular $\tilde R_T = R\cf_T$, i.e., the splicing operator is an exact unlearner in the sense of Definition~\ref{def:cf}.
\end{theorem}
\begin{proof}
Induction on $t$. \emph{Base} ($t=\tau-1$): $\tilde R_{\tau-1} = R_{\tau-1} = R\cf_{\tau-1}$ by the restore operation and Lemma~\ref{lem:prefix}. \emph{Step}: if $\tilde R_{t-1} = R\cf_{t-1}$ for some $t \ge \tau$, then
$\tilde R_{t} = F(\tilde R_{t-1},\tilde o_t;\theta) = F(R\cf_{t-1},\tilde o_t;\theta) = F(R\cf_{t-1},o\cf_t;\theta) = R\cf_{t}$,
where the second and third equalities use, respectively, (A1) to make $F$ single-valued (otherwise pointwise equality is not even well posed) and (A2)/(A3) to guarantee the step-$t$ tool observation attains $o\cf_t$ during replay.
\end{proof}

\begin{corollary}[Recomputation lower bound and optimality]
\label{cor:lb}
In the computational model where an operator may only (a) read the stored real trajectory $\{R_t\}_{0 \le t \le T}$, observation stream, and derived metadata (provenance graph, checkpoints, environment snapshots), or (b) invoke $F$ to advance a state, any exact unlearning operator must invoke $F$ at least $T-\tau+1$ times in the worst case. The splicing operator (Algorithm~\ref{alg:replay}) invokes it exactly $T-\tau+1$ times and is therefore \emph{optimal} under the conservative taint model; the gain over a full reset ($T$ invocations) is $T/(T-\tau+1)$.
\end{corollary}
\begin{proof}[Proof sketch]
\emph{Lower bound}: exactness requires outputting $R\cf_T$. When $z$ has nonzero influence, $R\cf_t \neq R_t$ for all $t \ge \tau$ in the worst case, so no state on the counterfactual suffix is stored and route (a) is unavailable; each invocation of route (b) advances the counterfactual trajectory by one step, and by Lemma~\ref{lem:prefix} the only stored state lying on it is at most $R_{\tau-1}$. Advancing from $R\cf_{\tau-1}$ to $R\cf_T$ takes $T-(\tau-1)$ invocations. \emph{Upper bound}: Algorithm~\ref{alg:replay} replays from $\tilde R_{\tau-1}$ in exactly $T-\tau+1$ steps, exact by Theorem~\ref{thm:splice}.
\end{proof}

Corollary~\ref{cor:lb} yields a testable prediction: forgetting cost equals the post-target \emph{suffix length} $T-\tau+1$, not the session length $T$---the later the target arrives, the closer forgetting is to free; the ablations below verify it. The bound counts sequential state \emph{transitions}, not autoregressively decoded tokens: content that is fixed within a replayed transition may be re-materialized by prefill, so token-level cost can fall below the transition count without contradicting the lower bound. (Since Theorem~\ref{thm:splice} makes equality a \emph{constructive} guarantee under deterministic decoding, the experiments emphasize efficiency and distributional consistency under stochastic decoding rather than treating agreement as a discovery.)

\subsection{From Theorems to System}
\label{sec:system}

Each theoretical object maps to a system component (Figure~\ref{fig:arch}), answering respectively \emph{where} to splice, \emph{what} to splice, and \emph{how}:

\textbf{(a) Provenance graph $G$ --- causal reachability, materialized (where).}
During execution we record artifact-level data flow: memory/tool field $\to$ prompt block $\to$ model turn $\to$ reply/plan $\to$ tool call $\to$ observation $\to$ summary write-back, each artifact carrying $(\mathrm{id}$, $\mathrm{type}$, $\mathrm{parents}$, $\mathrm{source\_ids}$, $\mathrm{token\_span}$, $\mathrm{turn}$, $\mathrm{committed})$. On a forget request, one reachability query returns the injection point $\tau$ and taint closure $I(z)$. We deliberately do \emph{not} attempt token-level attribution: the graph records only dependencies that actually occurred---conservative but certifiable (Lemma~\ref{lem:taint}ii).

\textbf{(b) Sparse checkpoint set $\mathcal{C}$ --- splice points, materialized (what).}
At semantic boundaries (session start, turn boundaries, before memory injections, consequential tool calls, and compactions) we register checkpoints holding only metadata (token offset, cache handle, environment-snapshot ID, prompt manifest)---no tensor copies. By Lemma~\ref{lem:prefix}, \emph{cropping is restoring}.

\textbf{(c) Sanitized replay --- re-enacting the counterfactual suffix (how).}
Read-only/deterministic tools are re-executed from the environment snapshot or replayed from recorded observations (realizing A2/A3); consequential tools are shadow-executed during replay with call-ID deduplication, so real side effects never fire twice.

Algorithm~\ref{alg:replay} assembles the pipeline: lines 1--3 are provenance queries (locate $\tau$, invalidate $I(z)$); lines 4--5 are the $O(1)$ checkpoint restore (Lemma~\ref{lem:prefix}); lines 6--10 re-enact the counterfactual suffix (the construction of Theorem~\ref{thm:splice}); total recomputation $T-\tau+1$ attains the bound of Corollary~\ref{cor:lb}.

\begin{algorithm}[t]
\caption{Provenance-Guided Selective Replay}
\label{alg:replay}
\begin{algorithmic}[1]
\REQUIRE runtime $R_T$, target $z$, provenance graph $G$, checkpoints $\mathcal{C}$, recorded observations $\{o_t\}_{t=1}^{T}$
\ENSURE counterfactual-equivalent runtime $R' = R\cf_T$
\STATE $s \leftarrow \mathrm{LocateSource}(G, z)$ \hfill $\triangleright$ source artifact of $z$
\STATE $\tau \leftarrow \mathrm{FirstEntry}(G, s)$ \hfill $\triangleright$ first entry boundary
\STATE $I(z) \leftarrow \mathrm{TaintClosure}(G, s)$ \hfill $\triangleright$ invalidate closure
\STATE $c^* \leftarrow \arg\max\{c \in \mathcal{C} : \mathrm{boundary}(c) < \tau\}$
\STATE $R \leftarrow \mathrm{Restore}(c^*)$ \hfill $\triangleright$ crop KV to $c^*$; load env snapshot
\FOR{$t = \mathrm{boundary}(c^*)+1 \dots T$}
  \STATE $\tilde o_t \leftarrow \mathrm{Sanitize}(o_t, z)$ \hfill $\triangleright$ remove $z$; keep the rest (A2)
  \STATE $\tilde o_t \leftarrow \mathrm{ReplayTools}(\tilde o_t)$ \hfill $\triangleright$ shadow exec + dedup (A3)
  \STATE $R \leftarrow F(R, \tilde o_t;\theta)$ \hfill $\triangleright$ regenerate derived artifacts
\ENDFOR
\RETURN $R' \leftarrow R$
\end{algorithmic}
\end{algorithm}

\section{Experiments}
\label{sec:experiments}

\subsection{Implementation and Setup}
\label{sec:impl}

\paragraph{Runtime harness.}
We implement this transition system on HuggingFace Transformers with explicit KV management. A \texttt{KVEngine} exposes the three primitives the theory needs: \texttt{prefill}, \texttt{generate}, and \texttt{crop}. Every block, turn, tool call, observation, and summary is a \texttt{RuntimeBlock} carrying the provenance tuple above; oracle source-ID propagation (a turn generated while the target is resident inherits its source ID) populates the \texttt{ArtifactGraph}, whose taint queries implement lines~1--3 of Algorithm~\ref{alg:replay}. A \texttt{CheckpointStore} keeps metadata-only handles at session start, turn boundaries, and before the target; \texttt{Restore} is a \texttt{crop} call, granularity ablated below.

\paragraph{Episodes.}
Each episode is a three-phase turn script---\emph{setup} (clean prefix), \emph{acquisition} (the target enters via memory injection or tool observation), \emph{contamination} ($\ge 1$ model turns folding the target into an answer and, per suite, a \emph{summary} and/or pending \emph{tool plan})---re-run with the target excluded to produce the counterfactual reference. Episodes are converted from \textbf{LongMemEval} \cite{wu2025longmemeval} ($n{=}100$, memory-injected facts), \textbf{ToolSandbox} \cite{lu2024toolsandbox} ($n{=}100$, tool-observed identifiers), and \textbf{AgentDojo} \cite{debenedetti2024agentdojo} ($n{=}80$; slack/workspace/banking/travel, 20 each); the tool-observation channel instantiates the detector-triggered case above: the target arrives inside a tool result the user never sees, and the forget request names the flagged observation. Future queries are task-shaped and solicit the target (e.g., ``schedule an appointment \emph{near my home}'').

\paragraph{Methods.}
All methods branch from the \emph{same} contaminated base state (cache cloned), so comparisons are paired: \textbf{B0} No-Forget; \textbf{B1} Memory-Delete (remove the persistent record, session untouched---what deployed stacks do); \textbf{B2} Forget-Instruction (append ``forget $z$''); \textbf{B3} Source-Redaction (drop the source block, keep derived artifacts); \textbf{B4} Sanitize-no-Replay (drop source + descendants, regenerate nothing); \textbf{B5} Full-Reset (the counterfactual reference $R\cf_T$); \textbf{B5$'$} Full-Reset $+$ prefix cache (control isolating how much of B7's saving a generic cache recovers); \textbf{B6} Sanitized-Rebuild (string-redact the transcript, rebuild); \textbf{B7} \textbf{Selective-Replay} (Algorithm~\ref{alg:replay}); \textbf{B8} FIDES-style IFC \cite{costa2025fides} (sink policy over identifier-type tool arguments, no state cleanup) and \textbf{B9} IFC + replay.

\paragraph{Models and decoding.}
Primary model: Llama-3.1-8B-Instruct on one RTX~4090; cross-family replication on Qwen2.5-7B-Instruct and Mistral-7B-Instruct-v0.3 \cite{grattafiori2024llama3,qwen2025qwen25,jiang2023mistral}. Main tables use temperature~0; stochastic audits use $k{=}5$ samples at sampling temperature $0.7$ under matched or independent seeds as noted ($T$ denotes session length throughout).

\paragraph{Metrics and audits.}
\emph{Leakage} is scored over the final answer, tool arguments, and memory write-backs: \emph{exact} (normalized string variants), \emph{action} (target in a tool argument, by sink class: routing / selection / free-text), and \emph{any}. \emph{CAD} (counterfactual action distance) scores over-deletion: tool-choice mismatch plus argument distance vs.\ the B5 reference. \emph{Utility} checks preserved non-target facts; \emph{efficiency} reports reused vs.\ recomputed tokens and latency. Since single-probe string matching is a lower bound, we add three audits, each disciplined by a measured false-positive floor (every probe also runs against B5; probes with nonzero floors are dropped---this excluded an LLM judge, floor 0.34): \textbf{Leak@probes} (six probes: task, direct, think, introspect, enumerate, cued-completion), \textbf{Leak@5} (5 stochastic samples), and a \textbf{behavioral-extraction} suite (below). Significance is paired throughout: exact McNemar (binary), Wilcoxon signed-rank (continuous).

\subsection{Deletion is Not Forgetting}
\label{sec:main-results}

\begin{table}[t]
\centering
\small
\setlength{\tabcolsep}{3pt}
\resizebox{0.85\columnwidth}{!}{%
\begin{tabular}{l ccc c c}
\toprule
 & \multicolumn{3}{c}{Any-leak $\downarrow$} & CAD $\downarrow$ & Recomp $\downarrow$ \\
\cmidrule(lr){2-4}
Method & LME & TS & AD & (avg) & (avg tok) \\
\midrule
B0 No-Forget          & 0.86 & 1.00 & 0.97 & 0.37 & 36 \\
B1 Memory-Delete      & 0.86 & 1.00 & 0.97 & 0.37 & 36 \\
B2 Forget-Instruction & 0.38 & 0.72 & 0.00$^{\dagger}$ & 0.39 & 55 \\
B3 Source-Redaction   & 0.84 & 0.94 & 0.39 & 0.33 & 200 \\
B4 Sanitize-no-Replay & 0.00 & 0.00 & 0.00 & 0.22 & 36 \\
B6 Sanitized-Rebuild  & 0.00 & 0.00 & 0.00 & 0.39 & 220 \\
B5 Full-Reset (ref)   & 0.00 & 0.00 & 0.00 & 0.00 & 1235 \\
B5$'$ +PrefixCache    & 0.00 & 0.00 & 0.00 & 0.00 & 202 \\
\textbf{B7 Selective-Replay} & \textbf{0.00} & \textbf{0.00} & \textbf{0.00} & \textbf{0.00} & \textbf{133} \\
\bottomrule
\end{tabular}}
\caption{Headline results (temperature 0; LME $n{=}100$, TS $n{=}100$, AD $n{=}80$; CAD/Recomp averaged over suites). Sessions run 12 turns past the target, half of them independent of it (\emph{indep} $=0.5$); Recomp counts tokens a method must recompute. $^{\dagger}$AD's task query never solicits the target; cf.\ B2 $=1.00$ in Table~\ref{tab:elicit}.}
\label{tab:main}
\vspace{-5pt}
\end{table}

Table~\ref{tab:main} makes the negative claim precise. \textbf{B1 equals B0 in every cell}: deleting the persistent record changes nothing---the target survives in answer, summary, plan, and cache ($p{=}1.0$ vs.\ B0). \textbf{B2} suppresses leakage but leaves the value resident (\emph{Elicitation}, below). \textbf{B3} leaks through derived artifacts: the model's summary and plan re-emit the target (0.84/0.94 any-leak). On ToolSandbox, 98\% of B0/B1 leaks flow through \emph{routing} arguments, steering side effects. \textbf{B4} achieves string-clean state by amputation but removes artifacts present in the counterfactual (CAD 0.22): over-deletion, not forgetting. \textbf{B8} (IFC only) blocks identifier sinks yet leaks at the B0 rate through free text (0.92--1.00); adding replay (\textbf{B9}) drops it to zero ($p{<}0.001$): state cleanup and flow control are orthogonal. \textbf{B7} matches B5 exactly (any-leak 0, CAD 0, agreement 1.0) while recomputing $9.3\times$ fewer tokens than a full reset (133 vs.\ 1235; $p{<}10^{-3}$) and $1.5\times$ fewer than B5$'$, the prefix-cached control that recovers the same prefill saving but re-decodes the whole suffix---that residual gap is what provenance buys, and it scales with \emph{indep}.


\subsection{Elicitation and Stochastic Audits}
\label{sec:elicit}

\begin{table}[t]
\centering
\small
\setlength{\tabcolsep}{4.5pt}
\begin{tabular}{l ccc c}
\toprule
 & \multicolumn{3}{c}{Leak@probes $\downarrow$} & Leak@5 $\downarrow$ \\
\cmidrule(lr){2-4}
Method & AD & LME & TS & TS \\
\midrule
B0 No-Forget          & 1.00 & 1.00 & 1.00 & 1.00 \\
B2 Forget-Instruction & 1.00 & 1.00 & 1.00 & 0.70 \\
B3 Source-Redaction   & 0.73 & 0.80 & 0.97 & 1.00 \\
B6 Sanitized-Rebuild  & --   & --   & --   & 0.00 \\
B8 FIDES-style IFC    & --   & --   & --   & 1.00 \\
B5 Full-Reset (floor) & 0.00 & 0.00 & 0.00 & 0.00 \\
\textbf{B7 Selective-Replay} & \textbf{0.00} & \textbf{0.00} & \textbf{0.00} & \textbf{0.00} \\
\bottomrule
\end{tabular}
\caption{Adversarial audits ($n{=}30$/suite). Leak@probes: leaked under \emph{any} of six probes against the same post-unlearning state; Leak@5: any of 5 samples at temperature $0.7$ (ToolSandbox). B5 $=$ measured false-positive floor.}
\label{tab:elicit}
\end{table}

Table~\ref{tab:elicit} shows why single probes mislead. \textbf{B2's Leak@probes is 1.00 on all three suites}: an introspection probe (``what were you told to forget?'') alone re-licenses the value at 0.87--1.00, and enumeration and chain-of-thought probes surface it where a direct question does not: the instruction leaves the value in state and commands silence. B3's paraphrased derivations yield 0.73--0.97; B7 sits at the B5 floor: zero under every probe, including cued completion, and zero on Leak@5. The audit also resolves Table~\ref{tab:main}'s AgentDojo anomaly: B2's task-probe 0.00 reflects a query that never solicits the target, not the method; safety claims must rest on Leak@probes.

\subsection{Behavioral Extraction: Influence Without Strings}
\label{sec:behavioral}

\begin{table}[t]
\centering
\small
\setlength{\tabcolsep}{4.2pt}
\begin{tabular}{l cccc}
\toprule
Method & avoid & residue & $p$ & says it \\
\midrule
B0 No-Forget          & 1.00 & $+0.68$ & $<.001$ & 0.00 \\
B1 Memory-Delete      & 1.00 & $+0.68$ & $<.001$ & 0.00 \\
B2 Forget-Instruction & 1.00 & $+0.68$ & $<.001$ & 0.00 \\
B3 Source-Redaction   & 0.80 & $+0.48$ & $<.001$ & 0.00 \\
B6 Sanitized-Rebuild  & 1.00 & $+0.68$ & $<.001$ & 0.00 \\
B5 Full-Reset (ref)   & 0.32 & ---     & ---     & 0.00 \\
\textbf{B7 Selective-Replay} & \textbf{0.32} & $\bm{+0.00}$ & 1.0 & 0.00 \\
\bottomrule
\end{tabular}
\caption{Behavioral extraction (30 preference episodes $\times$ 2 counterbalanced orders). The agent picks between two near-equivalent providers, one excluded by a \emph{revoked} preference; ``avoid'' $=$ rate of picking the other, ``says it'' $=$ string leakage of the revoked reason. The zero point is \emph{measured} (B5 $=$ 0.32), not assumed.}
\label{tab:behavioral}
\vspace{-5pt}
\end{table}

Every other leakage number here is a string match, so a method that stops \emph{saying} the target scores 0.00 whether or not it still shapes what the agent \emph{does}; Table~\ref{tab:behavioral} separates the two via a revoked \emph{preference} (an exclusion) the agent need never state to act on. \textbf{B0, B1, B2 and B6 all score 0.00 on string leakage yet act on the revoked preference in 100\% of episodes} (residue $+0.68$, $p{<}0.001$); B7 and B5$'$ match the reference ($+0.00$). A variant handing the redaction baselines an oracle (the excluded brand added to the forbidden list) is instructive: B6, which redacts \emph{every} artifact including the model-authored summary, drops to the floor (0.35, n.s.); B3, which keeps derived artifacts, stays at 0.80. The honest reading: string redaction works \emph{only if} it reaches every derived artifact \emph{and} the exact string is known in advance; a preference has neither property---the model paraphrases it into its own notes. Replay needs no such assumption.

\subsection{Exactness, Efficiency, and Generality}
\label{sec:efficiency}

\begin{figure}[t]
\centering
\includegraphics[width=1\columnwidth]{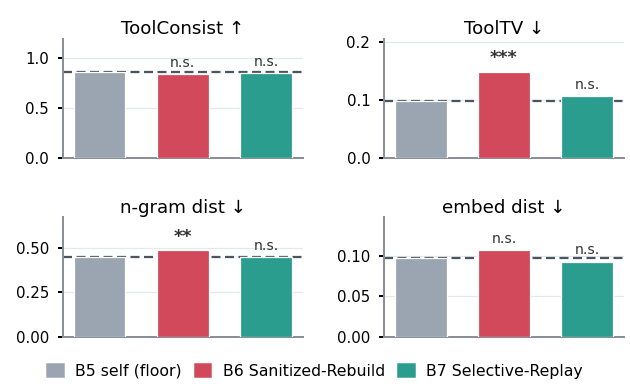}
\caption{Behavioral distribution under stochastic decoding (pooled $n{\approx}90$; $k{=}5$ samples at temperature $0.7$, independent seeds; dashed $=$ B5 self-sampling floor). Stars: Wilcoxon vs.\ floor ($^{**}p{<}.01$, $^{***}p{<}.001$). No B7 panel diverges detectably; B6's stars show the test has power.}
\label{fig:dist}
\vspace{-5pt}
\end{figure}

\paragraph{Exactness (Definition~\ref{def:cf}).}
Under matched seeds at temperature $0.7$, B7 reproduces the B5 reference \emph{token-for-token} in 450/450 draws ($\varepsilon{=}0$); B6 manages 1\%. Under independent seeds---the meaningful distributional test---no divergence from the B5 self-sampling floor is detectable for B7 on any of the four distances (tool consistency, tool TV, $n$-gram, embedding; Wilcoxon $p>0.5$ throughout), whereas the same test flags B6 (tool TV $p{<}10^{-3}$; Figure~\ref{fig:dist}). Non-significance is not proof of equivalence; B6 is the positive control showing the test has power here, and formal equivalence testing (TOST against a pre-registered margin) is future work.

\paragraph{The $T-\tau+1$ law (Corollary~\ref{cor:lb}).}
Sweeping the injection step $\tau$ from the first turn to the last, B5's recomputation stays flat while B7's descends linearly in $\tau$ ($R^2{=}1.000$ on all three suites; LongMemEval: 1218$\to$162 tokens vs.\ B5's constant 1330)---tracking the post-target suffix length, not the session length. The second axis is \emph{indep}, the fraction of post-target work causally independent of $z$---what moves B7 and B5$'$ apart. Real sessions carry such work in bulk: a skill card or document loaded and never used, a routine tool poll, an unrelated sub-task. Those turns are byte-identical counterfactually, so B7 re-\emph{prefills} them (parallel) and re-\emph{decodes} only the target-dependent remainder, while B5$'$ recovers the same prefill saving but re-decodes the whole suffix. B7's decode cost therefore falls linearly with \emph{indep} ($R^2{=}1.00$) while B5$'$ stays flat at 545--610 tokens; at full independence B7 decodes 114--147 tokens, 2.9--3.7 vs.\ 11.5--12.7\,s p50 (prefix caching alone: $1.01\times$). At \emph{indep} $=0$ they coincide exactly---the honest degenerate case---with any-leak and CAD at the B5 floor throughout. This saving relies on an attribution oracle stronger than Lemma~\ref{lem:taint}(ii)'s conservative model; it reduces token cost \emph{within} transitions, not their number, leaving Corollary~\ref{cor:lb} intact. B7 further provides cache-residency independence, snapshot restoration, and a certifiable audit trail.

\paragraph{Cross-model.}
The two load-bearing findings---B3 still leaks, B7 reaches the clean reference cheaply---reproduce on both other families (280 paired episodes each; B3 any-leak 0.84--1.00, B7 all-zeros, 2.9--4.3$\times$ savings). The one anomalous cell in the matrix (Llama's B3 $=$ 0.39 on AgentDojo) is model-specific (B3 $\geq$ 0.92 elsewhere)---a reason single-model unlearning evaluations mislead.

\subsection{Ablations}
\label{sec:ablations}

\begin{table}[t]
\centering
\footnotesize
\setlength{\tabcolsep}{4.5pt}
\begin{tabular}{l ccc c}
\toprule
 & \multicolumn{3}{c}{Any-leak $\downarrow$} & CAD $\downarrow$ \\
\cmidrule(lr){2-4}
Invalidation reach & AD & LME & TS & (avg) \\
\midrule
None (B0)                & 0.97 & 0.83 & 1.00 & 0.38 \\
Target span only         & 0.13 & 0.63 & 0.93 & 0.47 \\
Span + descendants (B4)  & 0.00 & 0.00 & 0.00 & 0.23 \\
+ replay (\textbf{B7})   & \textbf{0.00} & \textbf{0.00} & \textbf{0.00} & \textbf{0.00} \\
Full reset (B5)          & 0.00 & 0.00 & 0.00 & 0.00 \\
\bottomrule
\end{tabular}
\caption{Invalidation-boundary ablation ($n{=}30$/suite). Span-only excision is unsafe \emph{and} damages utility (CAD worse than no forgetting); safety needs the descendant closure, equivalence additionally needs replay.}
\label{tab:boundary}
\vspace{-10pt}
\end{table}

\paragraph{Invalidation boundary (Lemma~\ref{lem:taint} is not pessimism).}
Could one keep the suffix KV and excise just the target's span? Table~\ref{tab:boundary} says no: span-only excision still leaks in up to 93\% of episodes---the surrounding KV was \emph{computed while attending to} the target---and its CAD (0.47) is \emph{worse than no forgetting at all} (0.38), since deleting mid-context positions corrupts state the model misreads. Each escalation fixes one failure mode: the descendant closure (B4) zeroes leakage but leaves the runtime missing artifacts the counterfactual would have (CAD 0.23); only replay reaches equivalence. This is the empirical face of Lemma~\ref{lem:taint}, and the depth sweep gives the matching \emph{necessity} argument for provenance: B3 is clean at derivation depth~0 (any-leak 0.00) and degrades to 0.33--0.97 as an answer, plan, and summary stack on top; B7 holds any-leak $=$ CAD $=$ 0 at every depth.

\paragraph{Checkpoint granularity (an efficiency knob, not a correctness one).}
Sweeping checkpoint policies from \texttt{every\_turn} to \texttt{session\_start} (one checkpoint, no reusable pre-target prefix) leaves any-leak and CAD flat at zero: by Theorem~\ref{thm:splice}, replay from an earlier clean boundary is equally exact, merely longer. Granularity buys only prefill reuse ($\sim$1000 tokens, 0.1--0.15\,s) at negligible metadata cost ($\le$4.4\,KB); B7 keeps its latency advantage even at \texttt{session\_start} (7.8 vs.\ 12.2\,s), since it comes from provenance-guided re-\emph{prefilling}, which needs the artifact graph, not the checkpoint.

\paragraph{Target position (leakage is position-invariant; cost is not).}
Moving the injection early/middle/late leaves every method's leakage and CAD essentially unchanged---the target contaminates the session wherever it sits---while B7's cost ratio to a full reset falls from 0.91 to 0.14: the $T/(T-\tau+1)$ profile of Corollary~\ref{cor:lb}.

\section{Conclusion}

Stateful agents broke the equation between deleting a record and forgetting it: once read, information propagates into summaries, plans, and cached tensors that forget operations never touch. Defining forgetting as counterfactual equivalence---\emph{exactly} achievable at the runtime layer, where splicing the clean prefix to a replayed suffix costs $T-\tau+1$ transitions and no exact operator does better---we built Provenance-Guided Selective Replay, a cross-layer contract matching a full reset under every audit at a fraction of its cost.

\paragraph{Limitations.}
The guarantee covers reconstructible runtime state, not model parameters \cite{liu2025rethinking}, committed side effects (A3), or correlates of $z$; recomputation cannot drop below $T-\tau+1$. Assumption (A2) fixes post-$\tau$ observations, so replay covers snapshot-replayable tool returns and memory injections but not human turns that would have differed. Audits are string-based apart from the behavioral suite, and the distributional result is non-significance, not equivalence.

\bibliography{aaai2027}


\end{document}